\documentclass{article}

\usepackage{PRIMEarxiv}

\usepackage[utf8]{inputenc}      
\usepackage[T1]{fontenc}         
\usepackage[english]{babel}      
\usepackage{hyperref}            
\usepackage{url}
\usepackage{booktabs}
\usepackage{amsmath,amssymb,amsfonts}
\usepackage{amsthm}
\usepackage{nicefrac}
\usepackage{microtype}
\usepackage{multirow}
\usepackage{makecell}
\usepackage{indentfirst}
\usepackage{fancyhdr}
\usepackage{graphicx}
\usepackage{tikz}
\usepackage[square,authoryear]{natbib}
\graphicspath{{media/}}

\newtheorem{teorema}{Theorem}
\newtheorem{corollario}[teorema]{Corollary}
\newtheorem{lema}[teorema]{Lemma}

\newcommand{\val}{\operatorname{val}}

\title{Greedy approaches for Gold Grabbing on subclasses of split graphs
}

\author{
  Heitor Melo de Lucas Brandão \\
  Institute of Informatics \\
  Universidade Federal de Goiás \\
  Goiânia-GO, Brazil \\
  \texttt{heitormelo26@discente.ufg.br} \\
   \And
  Hebert Coelho da Silva \\
  Institute of Informatics \\
  Universidade Federal de Goiás \\
  Goiânia-GO, Brazil \\
  \texttt{hebert@ufg.br} \\
   \And
  Julliano Rosa Nascimento \\
  Institute of Informatics \\
  Universidade Federal de Goiás \\
  Goiânia-GO, Brazil \\
  \texttt{jullianonascimento@ufg.br} \\
}

\begin{document}
\maketitle

\begin{abstract}
The Gold Grabbing Game is a combinatorial game on vertex-weighted graphs in which two players alternately remove vertices while maintaining graph connectivity, aiming to maximize the total collected weight. Although the literature has primarily focused on strategies that guarantee victory, the question of optimality --- i.e., maximizing total gain --- remains less explored. In this work, we investigate the behavior of the greedy strategy in this setting. We show that this approach is not optimal for general split graphs, highlighting structural limitations of this class. On the other hand, we prove that for complete split graphs $CS_{(2,n)}$, the greedy strategy yields a sequence of moves that maximizes the game value. As a consequence, the first player does not lose when the number of vertices is even. These results contribute to a better understanding of the structural conditions that ensure the optimality of simple strategies in graph-based combinatorial games.
\end{abstract}


\section{Introduction}

The \emph{Gold Grabbing} (GG) game is a two-player combinatorial game defined on a connected graph $G = (V,E,w)$, where $w: V\to \mathbb{Z}^+$ assigns positive weights to the vertices. Alice and Bob take alternating turns removing vertices from the graph, with Alice moving first. In each round, the player to move must choose a vertex that is not a cut vertex and, by doing so, collects the corresponding weight (gold) of the vertex as score. The game ends when every vertex of the graph has been removed, and the winner is the player with the larger sum of collected weights.

The game was originally proposed by \cite{Winkler-2004} in the context of dividing a pizza between two players. Later, \cite{micek-2011} generalized the game to graphs, when it was named the \emph{Gold Grabbing Game}; moreover, they proved that Alice secures at least $1/4$ of the total weight of any even tree. That proof was extended in~\cite{seacrest-2012}, who proved that Alice secures at least half of the weight on even trees, that is, she does not lose the game, and conjectured that the same holds for every even bipartite graph. This conjecture remains open and has been the subject of further work establishing it for some subclasses of bipartite graphs: \cite{kmn-tree-2018} proved it for $K_{m,n}$-trees, and~\cite{blowupTree-2020} for \textit{blow-ups} of trees and cycles. Furthermore, regarding the complexity of the game, \cite{cibulka-2013} showed that GG is PSPACE-complete, evidencing the intrinsic computational hardness of the problem. In \cite{araujo2025teoria}, other combinatorial games on graphs are explored, which attests to the recent interest in this area.

That said, the state of the art on this game is predominantly focused on guaranteeing victory for one of the players, mainly Alice. A distinct question, and one that is far less explored, concerns \emph{optimality}, that is, which approach yields the sequence that optimizes a player's gain, maximizing $\val(G)$, explained in Section \ref{sec:preliminares}. In this sense, it is worth studying whether certain approaches enjoy such \emph{optimality}, such as the greedy one, which consists in choosing, at each step, the locally optimal option according to a fixed criterion, without considering future consequences, seeking to directly maximize the value obtained in each move. This question was partially explored by \cite{brandao2025goldgrabbing}, who showed that the greedy strategy is optimal for complete graphs $K_n$ and that for complete bipartite graphs $K_{m,n}$ it is a winning strategy.

A natural class in which to extend these results would be that of split graphs, which generalizes complete graphs and bears similarity to complete bipartite graphs. Nevertheless, as we show in Section~\ref{sec:gulosa}, the greedy approach is not optimal for split graphs in general. Because of this, we restrict our analysis to the subclass of complete split graphs $CS_{(2,n)}$. A \textit{complete split graph} $CS_{(p,n)}$ is a graph with $p+n$ vertices, composed of a clique $C$ of size $p$ and an independent set $I$ of size $n$, with $1 \leq p$ and $1 \leq n$, in which every vertex of the clique is adjacent to every vertex of the independent set.

In this work, we prove that the greedy strategy is optimal for the \textit{Gold Grabbing Game} on the complete split graphs $CS_{(1,n)}$ and $CS_{(2,n)}$. The graph $CS_{(2,n)}$ is particularly interesting because it is the smallest class of complete split graphs in which cut vertices arise naturally in the subgames, that is, as the rounds proceed new cut vertices appear (when a clique vertex is removed, the resulting graph is a star whose center is a cut vertex), which makes the analysis significantly more complex than for complete graphs.

The remainder of the paper is organized as follows. In Section \ref{sec:preliminares}, we present the definitions used throughout the work. Subsequently, in Section \ref{sec:gulosa}, we give a counterexample showing the failure of the greedy approach on general split graphs and why it fails; moreover, we prove the optimality of this approach for star graphs $K_{1,n}$, and we develop the preparatory lemmas and the main theorem for $CS_{(2,n)}$. Finally, in Section~\ref{sec:conclusão}, we present the conclusions and directions for future work.

\section{Preliminaries}
\label{sec:preliminares}
We define a graph $G$ as a pair of two finite sets $(V, E)$, where $V$ is the vertex set and $E$ is a subset of the unordered pairs of $V$. A graph $G$ is connected if, for any pair of vertices $u,v$, there is a path with endpoints $u$ and $v$. A vertex $v$ is called a \emph{cut vertex} if $G - v$ is not connected. Accordingly, a move $v$ is \emph{feasible} if $G - v$ remains connected, that is, if $v$ is not a cut vertex (or if $|V(G)| = 1$).

In the work of \cite{kmn-tree-2018}, a recurrence was defined that expresses the \emph{game value}, that is, the ``aggregate value'' of each feasible move, computed by brute force and defined below:
\begin{equation}\label{eq:val}
  \val(G) = \begin{cases}
    0, & \text{if } V(G) = \emptyset, \\
    \max\{w(v) - \val(G - v) \mid v \text{ is feasible}\}, & \text{if } V(G) \neq \emptyset.
  \end{cases}
\end{equation}

A move $v$ is \emph{optimal} if $\val(G) = w(v) - \val(G - v)$. If $\val(G) > 0$, Alice wins under optimal play; if $\val(G) = 0$, the game is a draw; if $\val(G) < 0$, Bob wins.

Figure~\ref{fig:exemplo-valG-P4} illustrates the computation of $\val(G)$ for a path $P_4: abcd$ with weights $w(a)=5$, $w(b)=4$, $w(c)=1$, $w(d)=3$. On the left, the decision tree makes explicit all possible sequences of moves, taking into account that, in each turn, the player chooses among the feasible vertices. On the right, the recursion tree generated by $\val(G)$ following Equation~\ref{eq:val} in a \emph{bottom-up} fashion: starting from the base cases $\val(\emptyset) = 0$, the values propagate up to the root, where each node stores $w(v) - \val(G - v)$ for the chosen vertex $v$. The value at the root indicates the advantage of the player to move under optimal play. In this example, the optimal sequence is $[a, b, d, c]$ with $\val(P_4) = 3$.

Note that the optimal sequence is obtained after computing $\val(G)$ for the graph and then greedily choosing the vertex whose ``aggregate value'' is largest, this being the optimal strategy for GG. The \emph{greedy strategy}, in turn, consists in each player, on their turn, choosing the available feasible vertex of largest weight:
\[
  v^* = \arg\max\{w(v) \mid v \text{ is feasible in } G\}.
\]

We say that the greedy strategy is \emph{optimal} on a class of graphs if, for every instance of the class, the sequence of moves produced by the greedy strategy coincides with at least one of the sequence(s) realizing $\val(G)$, since it is possible to have more than one vertex sequence with the same ``aggregate value''.

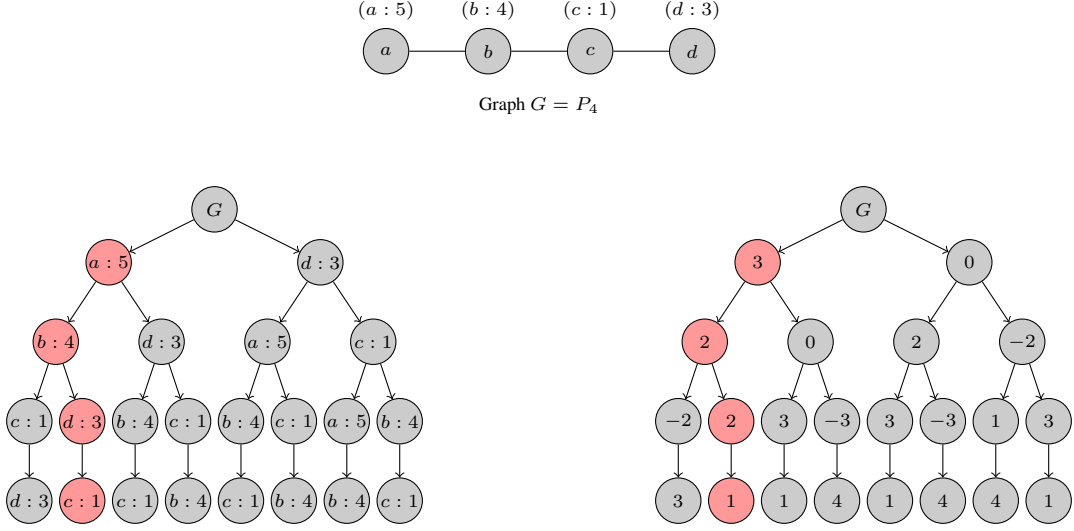
\begin{figure}[htbp]
\centering

\begin{tikzpicture}[scale=0.9, every node/.style={font=\scriptsize},
vertex/.style={circle,draw,fill=black!20,minimum size=0.6cm,inner sep=0pt}]
    \node[vertex] (v1) at (0,-1.2){$a$};
    \node[vertex] (v2) at (1.5,-1.2) {$b$};
    \node[vertex] (v3) at (3,-1.2) {$c$};
    \node[vertex] (v4) at (4.5,-1.2) {$d$};

    \node at (0,-0.6) {$(a:5)$};
    \node at (1.5,-0.6) {$(b:4)$};
    \node at (3,-0.6) {$(c:1)$};
    \node at (4.5,-0.6) {$(d:3)$};

    \draw (v1)--(v2)--(v3)--(v4);

    \node at (2.25,-2) {\scriptsize Graph $G = P_4$};
\end{tikzpicture}

\vspace{0.8cm}

\begin{minipage}{0.48\textwidth}
\centering
\begin{tikzpicture}[scale=0.7, every node/.style={font=\scriptsize},
vertex/.style={circle,draw,fill=black!20,minimum size=0.6cm,inner sep=0pt},
redvertex/.style={circle,draw,fill=red!40,minimum size=0.6cm,inner sep=0pt}]
\node[vertex] (a) at (0.5,-0.5) {$G$};

\node[redvertex] (b) at (-1.5,-1.5) {$a:5$};
\node[vertex] (c) at (2.5,-1.5) {$d:3$};
\draw[->] (a) -- (b);
\draw[->] (a) -- (c);

\node[redvertex] (d) at (-2.5,-3) {$b:4$};
\node[vertex] (e) at (-0.5,-3) {$d:3$};
\node[vertex] (f) at (1.5,-3) {$a:5$};
\node[vertex] (g) at (3.5,-3) {$c:1$};
\draw[->] (b) -- (d);
\draw[->] (b) -- (e);
\draw[->] (c) -- (f);
\draw[->] (c) -- (g);

\node[vertex] (h) at (-3,-4.5) {$c:1$};
\node[redvertex] (i) at (-2,-4.5) {$d:3$};
\node[vertex] (j) at (-1,-4.5) {$b:4$};
\node[vertex] (k) at (0,-4.5) {$c:1$};
\node[vertex] (l) at (1,-4.5) {$b:4$};
\node[vertex] (m) at (2,-4.5) {$c:1$};
\node[vertex] (o) at (3,-4.5) {$a:5$};
\node[vertex] (r) at (4,-4.5) {$b:4$};
\draw[->] (d) -- (h);
\draw[->] (d) -- (i);
\draw[->] (e) -- (j);
\draw[->] (e) -- (k);
\draw[->] (f) -- (l);
\draw[->] (f) -- (m);
\draw[->] (g) -- (o);
\draw[->] (g) -- (r);

\node[vertex] (s) at (-3,-6) {$d:3$};
\node[redvertex] (t) at (-2,-6) {$c:1$};
\node[vertex] (u) at (-1,-6) {$c:1$};
\node[vertex] (v) at (0,-6) {$b:4$};
\node[vertex] (w) at (1,-6) {$c:1$};
\node[vertex] (x) at (2,-6) {$b:4$};
\node[vertex] (y) at (3,-6) {$b:4$};
\node[vertex] (z) at (4,-6) {$c:1$};
\draw[->] (h) -- (s);
\draw[->] (i) -- (t);
\draw[->] (j) -- (u);
\draw[->] (k) -- (v);
\draw[->] (l) -- (w);
\draw[->] (m) -- (x);
\draw[->] (o) -- (y);
\draw[->] (r) -- (z);
\end{tikzpicture}
\end{minipage}
\hfill
\begin{minipage}{0.48\textwidth}
\centering
\begin{tikzpicture}[scale=0.7, every node/.style={font=\scriptsize},
vertex/.style={circle,draw,fill=black!20,minimum size=0.6cm,inner sep=0pt},
redvertex/.style={circle,draw,fill=red!40,minimum size=0.6cm,inner sep=0pt}]
\node[vertex] (a) at (0.5,-0.5) {$G$};

\node[redvertex] (b) at (-1.5,-1.5) {$3$};
\node[vertex] (c) at (2.5,-1.5) {$0$};
\draw[->] (a) -- (b);
\draw[->] (a) -- (c);

\node[redvertex] (d) at (-2.5,-3) {$2$};
\node[vertex] (e) at (-0.5,-3) {$0$};
\node[vertex] (f) at (1.5,-3) {$2$};
\node[vertex] (g) at (3.5,-3) {$-2$};
\draw[->] (b) -- (d);
\draw[->] (b) -- (e);
\draw[->] (c) -- (f);
\draw[->] (c) -- (g);

\node[vertex] (h) at (-3,-4.5) {$-2$};
\node[redvertex] (i) at (-2,-4.5) {$2$};
\node[vertex] (j) at (-1,-4.5) {$3$};
\node[vertex] (k) at (0,-4.5) {$-3$};
\node[vertex] (l) at (1,-4.5) {$3$};
\node[vertex] (m) at (2,-4.5) {$-3$};
\node[vertex] (o) at (3,-4.5) {$1$};
\node[vertex] (r) at (4,-4.5) {$3$};
\draw[->] (d) -- (h);
\draw[->] (d) -- (i);
\draw[->] (e) -- (j);
\draw[->] (e) -- (k);
\draw[->] (f) -- (l);
\draw[->] (f) -- (m);
\draw[->] (g) -- (o);
\draw[->] (g) -- (r);

\node[vertex] (s) at (-3,-6) {$3$};
\node[redvertex] (t) at (-2,-6) {$1$};
\node[vertex] (u) at (-1,-6) {$1$};
\node[vertex] (v) at (0,-6) {$4$};
\node[vertex] (w) at (1,-6) {$1$};
\node[vertex] (x) at (2,-6) {$4$};
\node[vertex] (y) at (3,-6) {$4$};
\node[vertex] (z) at (4,-6) {$1$};
\draw[->] (h) -- (s);
\draw[->] (i) -- (t);
\draw[->] (j) -- (u);
\draw[->] (k) -- (v);
\draw[->] (l) -- (w);
\draw[->] (m) -- (x);
\draw[->] (o) -- (y);
\draw[->] (r) -- (z);
\end{tikzpicture}
\end{minipage}

\caption{Example of $\val(P_4)$.}
\label{fig:exemplo-valG-P4}
\end{figure}

\subsection{The Graph $CS_{(2, n)}$}

Within the scope of this paper, we adopt the following terminology for $CS_{(2, n)}$: given $n \geq 1$, with $|V(CS_{(2, n)})| = n + 2$ vertices, $CS_{(2, n)}$ is composed of a clique $C = \{a, b\}$ (whose elements we call \emph{centers}) and an independent set $I = \{i_1, i_2, \ldots, i_n\}$ (whose elements we call \emph{independent vertices}). Formally:
\[
  V = \{a, b\} \cup \{i_1, \ldots, i_n\}, \qquad
  E = \{ab\} \cup \{ai_k : 1 \leq k \leq n\} \cup \{bi_k : 1 \leq k \leq n\}.
\]
Each independent vertex forms a triangle with the two centers. Equivalently, \(CS_{(2,n)}\) is obtained from the complete bipartite graph \(K_{2,n}\) by adding the edge \(ab\) between the vertices of the part of size 2, turning that part into a clique.

Figure \ref{fig:cs23_exemplo} illustrates an example of this class and will be used throughout this paper to ease the understanding of the lemmas and theorems.

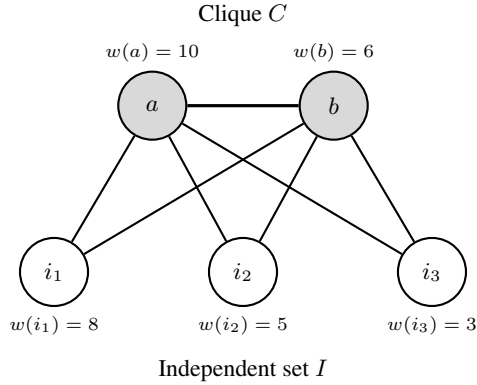
\begin{figure}[h]
\centering
\begin{tikzpicture}[
  every node/.style={draw, circle, minimum size=9mm, inner sep=0pt, font=\small},
  clique/.style={fill=gray!30},
  indep/.style={fill=white},
  thick
]
  \node[clique] (a) at (-1.2, 0) {$a$};
  \node[clique] (b) at ( 1.2, 0) {$b$};
  \node[indep] (i1) at (-2.5, -2.2) {$i_1$};
  \node[indep] (i2) at ( 0,   -2.2) {$i_2$};
  \node[indep] (i3) at ( 2.5, -2.2) {$i_3$};
  \draw[very thick] (a) -- (b);
  \draw (a) -- (i1); \draw (a) -- (i2); \draw (a) -- (i3);
  \draw (b) -- (i1); \draw (b) -- (i2); \draw (b) -- (i3);
  \node[draw=none, fill=none, font=\scriptsize] at (-1.2, 0.7) {$w(a)=10$};
  \node[draw=none, fill=none, font=\scriptsize] at (1.2, 0.7) {$w(b)=6$};
  \node[draw=none, fill=none, font=\scriptsize] at (-2.5, -2.9) {$w(i_1)=8$};
  \node[draw=none, fill=none, font=\scriptsize] at (0, -2.9) {$w(i_2)=5$};
  \node[draw=none, fill=none, font=\scriptsize] at (2.5, -2.9) {$w(i_3)=3$};
  \node[draw=none, fill=none, font=\footnotesize] at (0, 1.2) {Clique $C$};
  \node[draw=none, fill=none, font=\footnotesize] at (0, -3.5) {Independent set $I$};
\end{tikzpicture}
\caption{The graph $CS_{(2,3)}$.}
\label{fig:cs23_exemplo}
\end{figure}

In this example of Figure \ref{fig:cs23_exemplo}, all five vertices are feasible (each one belongs to at least one triangle). Computing $\val(G)$ by Equation~\ref{eq:val}, the first player has five options. Figure~\ref{fig:val_cs23} shows the value $w(v) - \val(G - v)$ for each possible choice of Alice in the first round:

\begin{figure}[h]
\centering
\begin{tikzpicture}[scale=0.85, every node/.style={font=\small},
  vertex/.style={circle, draw, minimum size=8mm, inner sep=0pt},
  opt/.style={circle, draw, fill=red!30, minimum size=8mm, inner sep=0pt}]
 
  \node[vertex] (root) at (0, 0) {$G$};
 
  \node[opt]    (a)  at (-5, -2) {$4$};
  \node[vertex] (b)  at (-2.5, -2) {$-4$};
  \node[vertex] (i1) at (0, -2) {$0$};
  \node[vertex] (i2) at (2.5, -2) {$0$};
  \node[vertex] (i3) at (5, -2) {$0$};
 
  \draw[->] (root) -- (a)  node[midway, above left, font=\scriptsize] {$a(10)$};
  \draw[->] (root) -- (b)  node[midway, left, font=\scriptsize] {$b(6)$};
  \draw[->] (root) -- (i1) node[midway, left, font=\scriptsize] {$i_1(8)$};
  \draw[->] (root) -- (i2) node[midway, right, font=\scriptsize] {$i_2(5)$};
  \draw[->] (root) -- (i3) node[midway, above right, font=\scriptsize] {$i_3(3)$};
 
  \node[draw=none, font=\scriptsize] at (-5, -3) {$K_{1,3}$};
  \node[draw=none, font=\scriptsize] at (-2.5, -3) {$K_{1,3}$};
  \node[draw=none, font=\scriptsize] at (0, -3) {$CS_{(2,2)}$};
  \node[draw=none, font=\scriptsize] at (2.5, -3) {$CS_{(2,2)}$};
  \node[draw=none, font=\scriptsize] at (5, -3) {$CS_{(2,2)}$};
 
\end{tikzpicture}
\caption{First level of the $\val(CS_{(2,3)})$ tree.}
\label{fig:val_cs23}
\end{figure}
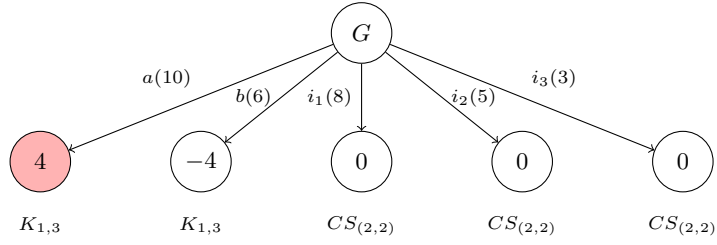

The maximum is $4$, obtained by choosing $a$ (weight $10$), whose removal produces the star $K_{1,3}$ with center $b$ (weight $6$) and leaves $\{8, 5, 3\}$, for which $\val(K_{1,3}) = 6$. The greedy strategy also chooses $a$ (the feasible vertex of largest weight), confirming that they coincide. The complete sequence of the game, both under the greedy and under the optimal strategy, is $[a,i_1,i_2,b,i_3]$, which grants Alice 18 in weight and 14 to Bob, hence $\val(G) = 4$.

In the following sections, we will verify that the preparatory lemmas reproduce these values.

\section{Greedy Approaches}
\label{sec:gulosa}

In this section we study the use of the greedy approach on two classes of graphs, aiming to draw conclusions about its optimality. For this section, we will use the fact that the greedy approach is optimal for complete graphs $K_n$, which is trivial, since in any vertex sequence every vertex is feasible.

\subsection{Counterexample for General split Graphs}\label{sec:contra_split}

The class of split graphs generalizes complete graphs and stars, on which we already know the greedy strategy is optimal (the proof for stars is in Section \ref{sec:guloso-estrela}). However, optimality does not extend to the split class as a whole. Figure~\ref{fig:contra_split} displays a split graph $G$ with clique $\{d, c\}$ and independent set $\{a,b\}$, where $b$ is adjacent only to~$c$ and $a$ is adjacent to $d$ and $b$.

\begin{figure}[h]
\centering
\begin{tikzpicture}[
  every node/.style={draw, circle, minimum size=8mm, inner sep=0pt, font=\small},
  clique/.style={fill=gray!30},
  indep/.style={fill=white},
  thick
]
  \node[indep] (a) at (0, 1.5) {$a$};
  \node[clique] (d) at (-1.3, -0.5) {$d$};
  \node[clique] (c) at (1.3, -0.5) {$c$};
  \node[indep] (b) at (3, 1.5) {$b$};

  \draw[very thick] (a) -- (d);
  \draw[very thick] (a) -- (c);
  \draw[very thick] (d) -- (c);

  \draw (b) -- (c);

  \node[draw=none, fill=none, font=\scriptsize] at (0, 2.2) {$w(a)=50$};
  \node[draw=none, fill=none, font=\scriptsize] at (-1.3, -1.2) {$w(d)=2$};
  \node[draw=none, fill=none, font=\scriptsize] at (1.3, -1.2) {$w(c)=100$};
  \node[draw=none, fill=none, font=\scriptsize] at (3, 2.2) {$w(b)=51$};

  \node[draw=none, fill=none, font=\footnotesize] at (0, -0.9) {Clique};
  \node[draw=none, fill=none, font=\footnotesize] at (3, 0.8) {Indep.};
\end{tikzpicture}
\caption{A split (non-complete) graph on which the greedy strategy fails.}
\label{fig:contra_split}
\end{figure}
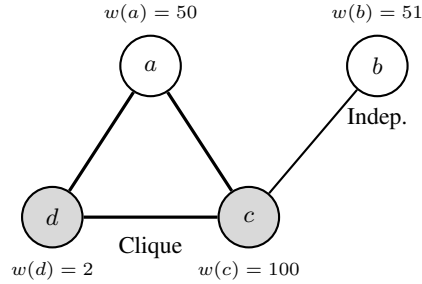

In this graph, the vertex $c$ is a cut vertex (its removal isolates $b$), so that the feasible vertices are $\{a, b, d\}$. The greedy strategy chooses $b$. After the removal of~$b$, the remaining graph is the triangle $\{a, d, c\}$, in which Bob takes $c$. The greedy sequence produces $[b, c, a, d]$, with Alice collecting $51 + 50 = 101$ and Bob collecting $100 + 2 = 102$, resulting in ``$\val_{\text{greedy}}\text{''} = -1$.

However, if Alice plays optimally and chooses $a$, the remaining graph is a star with center~$c$ and leaves $\{b, d\}$, in which $c$ is a cut vertex. Bob can take only $b$ or $d$. Since he will make the vertex $c$ available regardless of his choice, he takes $b$ for being the heavier one. Thus, Alice takes $c$ and Bob is left with $d$. In this example, the optimal sequence is $[a, b, c, d]$, with Alice collecting $50 + 100 = 150$ and Bob collecting $51 + 2 = 53$, resulting in $\val(G) = 97$. Whereas under the greedy strategy the sequence is $[b, c, a, d]$ with ``$\val_{\text{greedy}}\text{''} = -1$.

The difference is striking; the reason for the failure lies in the fact that $b$ is connected only to $c$, and not to the whole clique. By taking $b$ (the heaviest feasible vertex), Alice ``releases'' the triangle to Bob, who captures $c$. That is, the optimal move sacrifices immediate weight ($50 < 51$) in order to control the access to $c$.

This counterexample shows that the absence of complete adjacency between the independent set and the clique --- that is, the fact that the graph is split but not complete split --- is the source of the failure of the greedy strategy. And this is what motivates restricting the analysis to the class $CS_{(2,n)}$, in which complete adjacency is guaranteed and, therefore, so is the optimality of the greedy approach, proved in Section \ref{sec:gulosa_otima_CS}.

\subsection{Star Graph $K_{1,n}$}
\label{sec:guloso-estrela}
Let $G = K_{1,n}$ be a star graph with center $c$ and leaves $\{l_1, \ldots, l_n\}$, where $w(l_1) \geq \cdots \geq w(l_n)$. The fundamental structural property is that $c$ is a cut vertex as long as there are at least two leaves, so that the feasible vertices are exclusively the leaves until only two vertices remain.

\begin{teorema}\label{teo:gulosa_estrela}
The greedy strategy is optimal on $K_{1,n}$ for every $n \geq 1$.
\end{teorema}

\begin{proof}
As long as $|V(G)| \geq 3$, the center $c$ is a cut vertex and the feasible vertices are only the leaves. Since the leaves are mutually non-adjacent, the removal of any leaf does not change the feasibility of the others. Thus, the situation is structurally identical to playing GG on a complete graph $K_n$ formed by the leaves, where the greedy strategy is optimal: the players remove leaves in non-increasing order of weight.

To verify that this coincides with $\val(G)$, it suffices to apply $K_{1,n}$ to Equation \ref{eq:val} and observe that in the recursion $\val(K_{1,n}) = \max\{w(l_i) - \val(G - l_i)\}$, the leaf $l_1$ (of largest weight) satisfies:
\[
w(l_1) - \val(G - l_1) \geq w(l_k) - \val(G - l_k), \quad \forall\, k \geq 2,
\]
since $w(l_1) \geq w(l_k)$ and $\val(G - l_1) \leq \val(G - l_k)$ --- upon removing the heaviest leaf, the opponent faces a set of leaves with smaller weights. Applying this argument inductively at each round, the greedy strategy yields $\val(G)$.

When two vertices $\{c, l_n\}$ remain, both are feasible and the player to move takes the heavier one, which is trivially optimal; that is, if both players used the greedy approach, the sequence of vertices removed in GG will be the same as if they had used the optimal approach.
\end{proof}

\subsection{$CS_{(2,n)}$ Graphs}
\label{sec:gulosa_otima_CS}

In this subsection, we prove that the greedy strategy is optimal on $CS_{(2,n)}$ for every $n \geq 1$. The proof requires a series of preparatory lemmas about the structure of the subgames and the game value ($\val(G)$) on stars.

\smallskip
\noindent \textbf{Notation.} For a center $c$ with weight $w(c)$ and a set of leaves with weights $L = \{l_1, \cdots, l_k\}$, with $l_1 \geq \cdots \geq l_k$, we denote by $E(c, L)$ the star $K_{1,|L|}$ with center $c$ and leaves with weights given by $L$. We define the \emph{partial alternating sums} $A_j = \sum_{i=1}^{j} (-1)^{i-1} l_i$ for $j \geq 1$, with $A_0 = 0$. For $l_m \in L$, we write $L^{-m} = L \setminus \{l_m\}$.
\smallskip

\subsubsection{Structural lemmas}

The first two lemmas establish the basic properties that make the induction work: every vertex is feasible, and the resulting subgames belong to classes already covered by the induction hypothesis.

\begin{lema}[Initial Feasibility]\label{lema:cs_viavel}
In a $CS_{(2,n)}$, every vertex is feasible in the first round.
\end{lema}
\begin{proof}
Every vertex $v \in V(CS_{(2,n)})$ belongs to at least one triangle ($K_3$), since each independent vertex forms a triangle with the two centers, and the centers are shared by all triangles. As no vertex belonging to a cycle is a cut vertex, all of them are feasible.
\end{proof}

\begin{lema}[Subgames of $CS_{(2,n)}$]\label{lema:cs_sub}
In $CS_{(2,n)}$, with $n \geq 1$, removing an independent vertex results in $CS_{(2,n-1)}$ and removing a center results in a $K_{1,n}$.
\end{lema}
\begin{proof}
Removing an independent vertex $i_k$, the clique $\{a, b\}$ and the adjacencies with the remaining $n-1$ independent vertices are preserved, resulting in $CS_{(2,n-1)}$. Removing a center, say $a$, leaves $b$ adjacent to all $n$ independent vertices (mutually non-adjacent), forming $K_{1,n}$.
\end{proof}

\smallskip
\noindent \textbf{Example (Figure \ref{fig:cs23_exemplo}).} Initially, all five vertices are feasible: $i_1$ belongs to the triangle $\{a, b, i_1\}$, and each center belongs to three triangles. If we remove $i_3$, we obtain $CS_{(2,2)}$ with vertices $\{a(10), b(6), i_1(8), i_2(5)\}$. If we remove $a$, we obtain the star $K_{1,3}$ with center $b(6)$ and leaves $\{i_1(8), i_2(5), i_3(3)\}$.

\subsubsection{Lemmas on Stars $K_{1,n}$}

Given that Lemma \ref{lema:cs_sub} establishes that one of the subgames is a star, the following lemmas provide specific tools for analyzing the game value on that type of subgame.

\begin{lema}[Alternating Sums]\label{lema:alt_sum}
Let $l_1 \geq l_2 \geq \cdots \geq l_k > 0$ be a non-increasing sequence. Then:
\begin{enumerate}
    \item[(i)] $A_j \leq l_1$ for every $j \geq 0$.
    \item[(ii)] $A_{j-1} + l_j \leq l_1$ for every odd $j$ with $j \geq 1$.
\end{enumerate}
\end{lema}

\begin{proof}
\textbf{(i)} For even $j$: $A_j = l_1 - (l_2 - l_3) - (l_4 - l_5) - \cdots - l_j \leq l_1$, since each $(l_i - l_{i+1}) \geq 0$. For odd $j$: $A_j = (l_1 - l_2) + (l_3 - l_4) + \cdots + l_j \leq l_1$, since $l_j \leq l_1$.

\textbf{(ii)} $A_{j-1} + l_j = l_1 - (l_2 - l_3) - \cdots - (l_{j-1} - l_j) \leq l_1$, since each pair $(l_i - l_{i+1}) \geq 0$.
\end{proof}

The next lemma provides a closed formula for $\val(E(c, L))$, which will be used in the proof of the main theorem.

\begin{lema}[Star Value]\label{lema:val_estrela}
Let $E(c, L)$ be the star with center $c$ (weight $w(c)$) and leaves $L = \{l_1 \geq \cdots \geq l_k\}$. By the optimality of the greedy strategy on stars (Theorem~\ref{teo:gulosa_estrela}), the leaves are removed in non-increasing order and at the end the player to move chooses between $\{c, l_k\}$. Hence:
\[
\val(E(c, L)) = A_{k-1} + (-1)^{k-1} |l_k - w(c)|.
\]
The formula satisfies $\val(E(c,L)) = l_1 - \val(E(c, L \setminus \{l_1\}))$ for $|L| \geq 2$, with $\val(E(c, \{l\})) = |l - w(c)|$ and $\val(E(c, \emptyset)) = w(c)$.
\end{lema}

The two following lemmas relate the values of stars that differ by only one of the leaves (simulating the state of the game in which one must choose between two leaves as feasible moves). The central idea is that, given a center $c$ and leaves $L$, upon removing any leaf $l_m$ we obtain the star $E(c, L^{-m})$. The sum $\val(E(c,L)) + \val(E(c,L^{-m}))$ measures the ``total cost'' of having or not having $l_m$. Lemma~\ref{lema:soma_estrelas} provides an exact formula for this sum, and Lemma~\ref{lema:desig} bounds it from above.

\begin{lema}[Sum of Stars]\label{lema:soma_estrelas}
Let $L = \{l_1, \ldots, l_k\}$ (ordered non-increasingly) and let $c$ be the center with weight $w(c)$. Then:
\begin{enumerate}
    \item[(i)] $\val(E(c, L)) + \val(E(c, L^{-m})) = 2A_{m-1} + (-1)^{m-1} l_m$ \quad for $m < k$.
    \item[(ii)] $\val(E(c, L)) + \val(E(c, L^{-k})) = 2A_{k-2} + (-1)^{k-2} D$ \quad for $m = k$,
\end{enumerate}
where $D = l_{k-1} + |l_{k-1} - w(c)| - |l_k - w(c)|$.
\end{lema}

\begin{proof}
The strategy consists in defining $g(j) = \val(E(c, L_j)) + \val(E(c, L_j^{-m}))$, where $L_j = \{l_{j+1}, \ldots, l_k\}$, so that $g(0)$ is the desired sum. For $j < m-1$, the leaf $l_{j+1}$ is the maximum of both $L_j$ and $L_j^{-m}$ (since $l_m$, the leaf that differs, has index greater than $j+1$). Applying the recursion of Lemma~\ref{lema:val_estrela} to both and summing: $g(j) = 2l_{j+1} - g(j+1)$.

\textbf{(i) Case $m < k$:} In the base case $j = m-1$, the stars are $E(c, L_{m-1})$ (contains $l_m$) and $E(c, L_m)$ (does not contain $l_m$). Applying the recursion to the first one: $\val(E(c, L_{m-1})) = l_m - \val(E(c, L_m))$. Upon summing with $\val(E(c, L_m))$, a \emph{perfect cancellation} occurs: $g(m-1) = l_m$. This cancellation is the key property: the sum depends only on $l_m$ and on the preceding alternating sums, with no influence from the center~$c$. Solving down to $j = 0$: $g(0) = 2A_{m-1} + (-1)^{m-1} l_m$.

\textbf{(ii) Case $m = k$:} Here we remove $l_k$, the leaf that takes part in the final pair $\{c, l_k\}$, where the center interferes via $|l - w(c)|$. In the base case $j = k-2$, the stars are $E(c, \{l_{k-1}, l_k\})$ and $E(c, \{l_{k-1}\})$, with values $l_{k-1} - |l_k - w(c)|$ and $|l_{k-1} - w(c)|$. Unlike case~(i), the cancellation is \emph{not perfect}: $g(k-2) = D = l_{k-1} + |l_{k-1} - w(c)| - |l_k - w(c)|$, which depends on $w(c)$. Solving: $g(0) = 2A_{k-2} + (-1)^{k-2} D$.
\end{proof}

\smallskip
\noindent \textbf{Example (Figure \ref{fig:cs23_exemplo}).} Consider the star $E(b, L)$ with $L = \{8, 5, 3\}$ and $w(b) = 6$, obtained by removing the center $a$ of the example. By Lemma~\ref{lema:val_estrela}: $\val(E(b, L)) = A_2 + (-1)^2 |3 - 6| = (8 - 5) + 3 = 6$.

\emph{Case (i):} Removing $l_1 = 8$ (the largest leaf, $m = 1 < k = 3$), we obtain $L^{-1} = \{5, 3\}$ and $\val(E(b, \{5, 3\})) = 5 - |3 - 6| = 2$. The sum is $6 + 2 = 8$. By the formula: $2A_0 + (-1)^0 \cdot 8 = 0 + 8 = 8$. Note that the result does not depend on $w(b)$.

\emph{Case (ii):} Removing $l_3 = 3$ (the last leaf, $m = k = 3$), we obtain $L^{-3} = \{8, 5\}$ and $\val(E(b, \{8, 5\})) = 8 - |5 - 6| = 7$. The sum is $6 + 7 = 13$. By the formula: $D = 5 + |5 - 6| - |3 - 6| = 3$, hence $2A_1 + (-1)^1 \cdot 3 = 16 - 3 = 13$. Here the result depends on $w(b) = 6$, confirming the interference of the center.

\begin{lema}[Main Inequality]\label{lema:desig}
Let $l_1 \geq \cdots \geq l_k > 0$ be leaves, $c$ the center with weight $w(c)$, and $w(a)$ a value with $w(a) \geq l_i$ for every $i$ and $w(a) \geq w(c)$. Then, for every $m \in \{1, \ldots, k\}$:
\[
2w(a) - l_m \geq \val(E(c, L)) + \val(E(c, L^{-m})).
\]
\end{lema}

\begin{proof}
\textbf{Case $m < k$.} By Lemma~\ref{lema:soma_estrelas}(i), we need $2w(a) - l_m \geq 2A_{m-1} + (-1)^{m-1} l_m$. If $m$ is odd, this is equivalent to $w(a) \geq A_{m-1} + l_m \leq l_1$ (Lemma~\ref{lema:alt_sum}(ii)). If $m$ is even, it is equivalent to $w(a) \geq A_{m-1} \leq l_1$ (Lemma~\ref{lema:alt_sum}(i)). In both, $w(a) \geq l_1$ guarantees the inequality.

\textbf{Case $m = k$.} By Lemma~\ref{lema:soma_estrelas}(ii), we need $2w(a) - l_k \geq 2A_{k-2} + (-1)^{k-2} D$, with
\[
D = \begin{cases}
2l_{k-1} - l_k & \text{if } w(c) \leq l_k, \\
2l_{k-1} - 2w(c) + l_k & \text{if } l_k < w(c) \leq l_{k-1}, \\
l_k & \text{if } w(c) > l_{k-1}.
\end{cases}
\]

If $k$ is even ($(-1)^{k-2} = 1$): we check $D + l_k \leq 2l_{k-1}$ in the three scenarios (immediate). Hence $2A_{k-2} + D + l_k \leq 2A_{k-2} + 2l_{k-1} = 2A_{k-1} \leq 2l_1 \leq 2w(a)$.

If $k$ is odd ($(-1)^{k-2} = -1$): we check $-D + l_k \leq 0$ in the three scenarios (immediate). Hence $2A_{k-2} - D + l_k \leq 2A_{k-2} \leq 2l_1 \leq 2w(a)$.
\end{proof}

\smallskip
\noindent \textbf{Example (Figure \ref{fig:cs23_exemplo}).} We verify Lemma~\ref{lema:desig} with the data of the example. Taking $w(a) = 10$ as the global maximum, center $c = b$ with $w(b) = 6$, and $L = \{8, 5, 3\}$:

For $m = 1$ (leaf $l_1 = 8$): $2 \cdot 10 - 8 = 12 \geq \val(E(b, \{8, 5, 3\})) + \val(E(b, \{5, 3\})) = 6 + 2 = 8$.

For $m = 3$ (leaf $l_3 = 3$): $2 \cdot 10 - 3 = 17 \geq \val(E(b, \{8, 5, 3\})) + \val(E(b, \{8, 5\})) = 6 + 7 = 13$.

This inequality will be applied directly in Cases~3 and~4 of Theorem~\ref{teo:cs_gulosa}.

\subsubsection{Main Theorem}

With the lemmas established, we can now prove the central result of this work.

\begin{teorema}[Optimality of the Greedy Strategy on $CS_{(2,n)}$]\label{teo:cs_gulosa}
For every $n \geq 1$, the greedy strategy is optimal on $CS_{(2,n)}$ in GG.
\end{teorema}

\begin{proof}
Strong induction on $n$. Base: $n = 1$ ($K_3$, trivial). Inductive step ($n \geq 2$): we assume (IH1) the greedy strategy is optimal on $CS_{(2,j)}$ for $j < n$ and (IH2) the greedy strategy is optimal on $K_{1,m}$ for every $m \geq 1$.

By Lemma~\ref{lema:cs_viavel}, every vertex is feasible. The greedy strategy takes $v^*$ (largest weight). By Lemma~\ref{lema:cs_sub}, $G - v^*$ is $CS_{(2,n-1)}$ or $K_{1,n}$, both covered by the IH. It remains to show that $w(v^*) - \val(G - v^*) \geq w(u) - \val(G - u)$ for every $u \neq v^*$.

\medskip

\textbf{Case 1: $v^*$ and $u$ are independent vertices.}
Both subgames are $CS_{(2,n-1)}$. By IH1, Bob takes $v^*$ in $G - u$ and $u$ in $G - v^*$. After these moves, both converge to $G' = CS_{(2,n-2)}$, hence:
\[
\val(G - u) - \val(G - v^*) \geq w(v^*) - w(u) \geq 0.
\]

\textbf{Case 2: $v^*$ and $u$ are centers.}
The subgames are $E(u, L)$ and $E(v^*, L)$, with $L = \{l_1 \geq \cdots \geq l_n\}$ (the independent vertices). By Lemma~\ref{lema:val_estrela}, the difference of the values lies only in the last term. The desired inequality is rewritten as:
\[
w(v^*) - w(u) \geq (-1)^{n-1}\bigl(|l_n - w(u)| - |l_n - w(v^*)|\bigr).
\]
Since $w(v^*) \geq l_n$ always, two scenarios remain. If $l_n \leq w(u)$: the right-hand side is $(-1)^{n-1}(w(u) - w(v^*))$, and the inequality holds since $|w(u) - w(v^*)| \leq w(v^*) - w(u)$. If $w(u) \leq l_n$: from the inequalities $w(u) \leq l_n \leq w(v^*)$ we obtain $|2l_n - w(u) - w(v^*)| \leq w(v^*) - w(u)$, which guarantees the inequality for both parities of $n$.

\medskip

\textbf{Case 3: $v^*$ is a center and $u$ is an independent vertex.}
Let $c'$ be the other center. After removing $v^*$: the star $E(c', L)$. After removing $u$: $CS_{(2,n-1)}$, in which $v^*$ (the global maximum) is present, so Bob takes $v^*$ and $E(c', L^{-u})$ remains. Substituting $\val(CS_{(2,n-1)}) = w(v^*) - \val(E(c', L^{-u}))$, the desired inequality becomes:
\[
2w(v^*) - w(u) \geq \val(E(c', L)) + \val(E(c', L^{-u})),
\]
which is exactly Lemma~\ref{lema:desig} with $w(v^*)$ as the global maximum.

\medskip

\textbf{Case 4: $v^*$ is an independent vertex and $u$ is a center.}

This is the most delicate case. The central difficulty is that, upon removing $v^*$ (the global maximum), the resulting subgame $CS_{(2,n-1)}$ no longer has a vertex of dominant weight, and we do not know \emph{a priori} which vertex Bob will optimally choose --- it may be a center or an independent vertex. In the previous cases, the identity of Bob's move was directly determined; here, we need to analyze three subcases.

Let $c'$ be the other center. After removing $u$ (a center): we obtain the star $E(c', L)$, where $L$ contains all $n$ independent vertices, including $v^*$. Since $v^*$ is the heaviest leaf ($w(v^*) \geq w(c')$), Bob takes $v^*$ and $\val(E(c', L)) = w(v^*) - \val(E(c', L^{-v^*}))$. After removing $v^*$ (an independent vertex): we obtain $CS_{(2,n-1)}$. The desired inequality reduces to:
\begin{equation}\label{eq:caso4}
2w(v^*) - w(u) \geq \val(CS_{(2,n-1)}) + \val(E(c', L^{-v^*})).
\end{equation}

We analyze according to Bob's optimal move in $CS_{(2,n-1)}$:

\emph{Bob takes $u$:} $E(c', L^{-v^*})$ remains, hence $\val(CS_{(2,n-1)}) = w(u) - \val(E(c', L^{-v^*}))$ and the right-hand side of~(\ref{eq:caso4}) simplifies to $w(u)$. The inequality $2w(v^*) - w(u) \geq w(u)$ holds since $w(v^*) \geq w(u)$.

\emph{Bob takes $c'$:} Analogous to the previous subcase by symmetry between the centers.

\emph{Bob takes an independent vertex $p_j$:} This is the non-trivial subcase. $CS'_{(2,n-2)}$ remains, hence $\val(CS_{(2,n-1)}) = w(p_j) - \val(CS'_{(2,n-2)})$. In the graph $CS'_{(2,n-2)}$, taking $u$ is a valid move (Lemma~\ref{lema:cs_viavel}), which provides a \emph{lower bound}:
\[
\val(CS'_{(2,n-2)}) \geq w(u) - \val(E(c', L^{-v^*,-j})).
\]
Substituting:
\[
\val(CS_{(2,n-1)}) + \val(E(c', L^{-v^*})) \leq w(p_j) - w(u) + \val(E(c', L^{-v^*})) + \val(E(c', L^{-v^*,-j})).
\]
By Lemma~\ref{lema:desig}: $\val(E(c', L^{-v^*})) + \val(E(c', L^{-v^*,-j})) \leq 2w(v^*) - w(p_j)$. Therefore, the right-hand side is at most $2w(v^*) - w(u)$, satisfying~(\ref{eq:caso4}).

\medskip
In all cases, $v^*$ is optimal. By the principle of induction, the greedy strategy is optimal on $CS_{(2,n)}$ for every $n \geq 1$.
\end{proof}

\begin{corollario}
If $|V(CS_{(2,n)})| = n + 2$ is even, then $\val(CS_{(2,n)}) \geq 0$: Alice does not lose with the greedy strategy.
\end{corollario}
\begin{proof}
By Theorem~\ref{teo:cs_gulosa}, the greedy strategy is optimal. Since $|V|$ is even, Alice and Bob take $(n+2)/2$ vertices each. In the greedy sequence, Alice chooses before Bob in each pair of turns, always taking the heaviest available vertex. Hence $w(\text{Alice}_i) \geq w(\text{Bob}_i)$ for every $i$, and $\sum w(\text{Alice}) \geq \sum w(\text{Bob})$.
\end{proof}

\section{Conclusion}
\label{sec:conclusão}
In this work, we investigated the optimality of the greedy strategy for the Gold Grabbing Game on split graphs. We showed by means of a counterexample that the greedy approach is not optimal for general split graphs, and that this happens due to the absence of complete adjacency between the independent set and the clique, allowing the greedy move to ``release'' high-value vertices to the opponent. In contrast, we formally proved that the greedy strategy is optimal for the subclass of complete split graphs $CS_{(2,n)}$ for every $n \geq 1$. That is, the sequence of moves produced by the greedy strategy coincides with the sequence that maximizes $\val(G)$, a result that goes beyond the mere guarantee of victory, this result being extended to star graphs $K_{1,n}$. As a direct consequence, we obtained that Alice does not lose under optimal play on $CS_{(2,n)}$ with an even number of vertices, confirming the conjecture of \cite{seacrest-2012} for this class.

As future directions, we highlight: (i) the extension of the proof to general complete split graphs, where preliminary results indicate that an exchange argument settles the cases $c \geq 3$; (ii) the investigation of the optimality of the greedy strategy on threshold graphs, which are particular cases of split graphs and for which the greedy strategy appears to be optimal; and (iii) the precise characterization of which structural properties of a graph ensure or prevent the optimality of the greedy strategy, given that it fails on paths, cycles and general splits, but is optimal on complete graphs, stars and $CS_{(2,n)}$.

\bibliographystyle{plainnat}
\bibliography{references}

\end{document}